\documentclass[letterpaper, 10 pt, conference]{ieeeconf} 

\usepackage{amssymb}
\usepackage{amsmath}
\usepackage{amsfonts}
\usepackage{graphicx}

\usepackage{caption}
\usepackage{subcaption}
\allowdisplaybreaks

\newtheorem{theorem}{Theorem}

\newtheorem{algorithm}[theorem]{Algorithm}

\newtheorem{corollary}[theorem]{Corollary}

\newtheorem{proposition}[theorem]{Proposition}
\newtheorem{remark}[theorem]{Remark}

\usepackage{algorithm}
\usepackage[noend]{algpseudocode}

\newcommand{\colvec}[2][.8]{%
  \scalebox{#1}{%
    \renewcommand{\arraystretch}{.8}%
    $\begin{bmatrix}#2\end{bmatrix}$%
  }
}

\usepackage{epstopdf}

\title{On the Computation of Worst Attacks: a LP Framework}
\author{Nabil H. Hirzallah and Petros G. Voulgaris}

\begin{document}

\maketitle
\thispagestyle{empty}
\pagestyle{empty}

\begin{abstract}
We consider the problem of false data injection attacks modeled as additive disturbances in various parts of a general LTI feedback system and derive necessary and sufficient conditions for the existence of stealthy unbounded attacks. We also consider the problem of characterizing the worst, bounded and stealthy attacks.  This problem involves a maximization of a convex function subject to convex constraints, and hence, in principle, it is not easy to solve.  However, by employing a $\ell_\infty$ framework, we show how tractable Linear Programming (LP) methods can be used to obtain the worst attack design.  Moreover, we provide a controller synthesis iterative method to minimize the worst impact of such attacks.

\end{abstract}

\section{Introduction}
Advancements in communication, sensing and computing technologies allowed the control of physical systems or plants to be implemented over networks (cyber space), leading to the creation of ``cyber-physical" systems. Such systems are found in many applications including the smart grid and vehicle control units. However, the interaction between cyber systems and physical systems introduced security challenges that can be exploited by malicious agents. It has been shown through real world incidents and research paper simulations that stealthy attacks can be carefully designed to cause significant damage in control systems. Therefore, it is very important to research the security vulnerabilities of cyber-physical systems, and find solutions that guarantee the stability and resiliency of control systems under different attack scenarios. \\
Recent work on security of cyber-physical systems from a control-theoretic perspective has been focused on the characterization of feasible attacks and proposing ways for detection and/or improving the resiliency of the control system subject to such attacks. The type of attacks studied can be generally split into two categories: static attacks (attacks that do not take into account the dynamics of the system and/or do not affect the states of the system directly) or dynamic attacks. Attacks under each category can be classified as stealthy or not stealthy depending on the assumptions and the detection methods used. Examples of static attacks include attacks on the power system state estimators \cite{sandberg1}, where a carefully designed bias can be added to the sensor measurements without being detected by the commonly used statistical detection methods. Another work on static attacks is by \cite{tabuada} and \cite{hespanha} where they showed that the states of the system cannot be accurately reconstructed if half of the sensors are attacked. Both papers propose computationally intensive methods to reconstruct the states when less than half of the sensors are attacked. Their work was extended by \cite{fraz} where the authors provide a framework to to reconstruct the states that is robust to additive and multiplicative errors. On the other hand, research work related to dynamic attacks include \cite{dual} where the authors provide necessary and sufficient conditions for the existence of unbounded stealthy actuator and/or sensor attacks. In addition they proposed dual rate control to detect unbounded stealthy actuator attacks (zero dynamics attacks). In \cite{sinopoli} the authors inject a random signal (unknown to the attacker) into the system to detect replay attacks at the expense of increasing the cost of the LQG controller. In \cite{sandberg2} the case for finding the worst bias constant (steady state) attack has been considered and a tractable procedure to compute it has been developed where the energy of the detection signal was considered as a measure of stealthiness. In \cite{insok} coordinated actuator and sensor attacks are computed that create unbounded expectation of the estimation error while keeping the residual of the KF detector bounded. In \cite{lqg} optimal attacks are computed on a LQG systems that minimize the K-L divergence between the true and falsified state estimates such that the attack impact is above a specified a limit, showing that the optimal attacks are additive white noise. In \cite{wu} optimal actuator attacks are designed using the minimum principle that maximizes a quadratic cost related to the error between the healthy (un-attacked) system and the attacked system while minimizing the attack cost, without including any stealthiness requirement.\\
In this work, we consider signal attacks where the general problem from the attacker's perspective is to find the attack input $d =
\{d(k)\}$ so that it is stealthy while inflicting the maximum damage on the performance variable $z=\{z(k)\}$. We showed in our previous work \cite{dual} that unbounded attacks for LTI systems are related to the unstable zeros and/or poles of the open loop system. However, in this paper we consider the problem of characterizing the worst, bounded and stealthy attacks. This problem involves a maximization of a convex function subject to convex constraints. We propose different attack resource constraints to make the problem more practical. More specifically, we assume that the attacker has a finite time window $\{0,1,\dots,t_a\}$ to attack the system and inflict the maximum damage before the attack is over, and we attempt to solve the following three attack scenarios: 
Scenario 1 : Attacker can attack in a finite time window up to $t=t_a$, his goal is to inflict the maximum damage anywhere (before or after $t_a$) while remaining stealthy for all $t$.\\
Scenario 2 : Attacker can attack in a finite interval up to $t=t_a$, his goal is to inflict the maximum damage anywhere (before or after $t_a$) while remaining stealthy for $t\leq t_a$ (does not care if detected after the attack is over).\\
Scenario 3 : Attacker can attack in a finite interval up to $t=t_a$, his goal is to inflict the maximum damage at $t\leq t_a$ while remaining stealthy for $t\leq t_a$.\\
We show that by employing a $\ell_\infty$ framework, tractable Linear Programming (LP) methods can be used to compute the worst attack for the above three scenarios. Our work is closely related to \cite{sandberg2}, \cite{insok} and \cite{lqg}. However, we don't assume a constant $d$ such as in \cite{sandberg2} where they assume the system is in steady state. In addition, the work in \cite{insok} and \cite{lqg} relate to either a specific detection method (e.g. residual detectors) or to a specific controller in use. We plan to investigate these problems in a more general input-output fashion that does not depend on the particular controller used.\\
In the second part of this paper, we build on the worst attack design problem and provide a $K$-$d$ controller synthesis iterative method to minimize the performance cost without increasing the impact of the worst attack. Each iteration is a LP and alternates between finding the worst attack $d$ for a given controller $K$, and finding the next $K$ that minimizes the performance cost while keeping a non-increasing upper bound on the worst case impact inflicted by $d$.   

Some standard notation we use is as follows: $\mathbb{Z}_{+}$, $\mathbb{R}^{n}$, $\mathbb{C}^{n}$ and $\mathbb{R}^{n\times m}$ denote the sets of non-negative integers, $n$-dimensional
real vectors, $n$-dimensional
complex vectors and $n\times m$ dimensional real matrices, respectively. For any $\mathbb{R}^{n}$ or $\mathbb{C}^{n}$ vector $x$ we denote $x'$ its transpose and $|x|:=\max_i\sqrt{x_i^2}$ where  $x'=\left[ x_{1},x_{2},...,x_{n}\right]$; for a sequence of real $n$-dimensional
vectors, $x=\{x(k)\}_{k\in\mathbb{Z}_{+}} $ we denote $||x||_\infty:=\sup_k|x(k)|$; for a sequence of real $n\times m$ dimensional real matrices $G=\{G_k\}_{k\in\mathbb{Z}_{+}} $ we denote its $z$-transform $G(z):=\sum_{k=0}^{\infty}G_k z^{-k};$ and if viewed as the pulse response of the LTI system $G$ then $||G||_1=\sup_{||x||_\infty\le1}||Gx||_\infty$.  Finally, we will call system $G$ ``tall" when $y=Gu$ and the dimension of the output vector $y$ is at least equal to the  dimension of the input $u$; otherwise we will call $G$ ``fat".  We will also be using the standard notions for zeros and coprime factorizations of a LTI system $G$ (e.g., \cite{antsaklisBook,ZhouDoyleGloverBook,dahlehb}.)


\section{Problem Setup}


We consider the case of a general signal attack $d$ on a closed loop system of Figure \ref{general}.  Let $\Phi(K)$ describe the effect of $d$ on the performance variable $z$ and on the monitoring signal $\psi$, i.e. let  $\Phi=\begin{bmatrix} \Phi_{zd}\\
\Phi_{\psi d}\end{bmatrix}=:d\mapsto\begin{bmatrix} z\\
\psi
\end{bmatrix} $.  The monitoring signal $\psi$ consists of the measured output $y$ and the control signal $u$; it can however contain any other information that is recorded and measured, e.g., reference inputs.  In this setup, we assume that there may be other external disturbances and noise inputs which are ``normal", i.e., not malicious attackers, which are not shown in the figure.  Also, all the formulation deals with  discrete-time systems and signals.

The attacker's goal can be stated in general as
\begin{equation}
\!\begin{aligned}
&\max\limits_{d} \left\|z\right\|_\infty\\
&\text{s.t.}~ \left\|\psi\right\|_\infty\leq\theta,\\
\end{aligned}
\label{eqn:attack}
\end{equation}
where $\theta$ is an alarm threshold, associated with the afore mentioned normal set of disturbances.  In our previous work \cite{dual}, we established exact conditions for stealthiness of unbounded actuator and sensor attacks which can totally destroy the system.  These attacks are ultimately related to the open loop plant $P$, and  for LTI systems in particular, to the non-minimum phase (unstable) zeros and unstable poles of $P$.  We note, as pointed in \cite{dual}, that unstable zeros can also be due to the sampled data implementation of controllers.


\begin{figure}[t]
\centering
\includegraphics
[width=2in]
{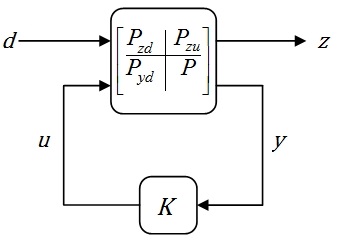}
\caption{General setup of input-output maps.}
\label{general}
\end{figure}

In this general setup of Figure \ref{general}, we elaborate on the existence of stealthy unbounded attacks using an input-output approach.  In particular, considering a left coprime factorization (\cite{antsaklisBook,ZhouDoyleGloverBook,dahlehb}) for the part of the generalized system that connects inputs to the measured output $y=P_{yd}d+Pu$ in the open loop, we have
$$[P_{yd}~P]=\tilde{M}^{-1}[\tilde{N}_{yd}~\tilde{N}]$$

Using a left coprime factorization for the stabilizing controller $K=YX^{-1}$  we can express
\begin{equation*}
\psi=\left[
\begin{array}{c}
y \\
u%
\end{array}%
\right]=\left[
\begin{array}{c}
X \\
Y%
\end{array}%
\right]W^{-1}\tilde{M}P_{yd}d=\left[
\begin{array}{c}
X \\
Y%
\end{array}%
\right]W^{-1}\tilde{N}_{yd}d.  
\end{equation*}
where  $W=\tilde{M}X-\tilde{N}Y$.  Since $W$ is stable and, by stability of the closed loop, has a stable inverse $W^{-1}$ we have that the detectability of $d$ depends on the unstable zeros of $\tilde{N}_{yd}$: unbounded stealthy attacks $d$ are possible if and only if $\tilde{N}_{yd}$  has unstable zeros.

For actuator only attacks $$P_{yd}=P=\tilde{M}^{-1}\tilde{N}\Longrightarrow\tilde{N}_{yd}=\tilde{N}$$
while for sensor only attacks $$P_{yd}=I=\tilde{M}^{-1}\tilde{M}\Longrightarrow\tilde{N}_{yd}=\tilde{M}.$$

Hence, this shows how the unstable zeros of $P$ (which are the unstable zeros of $\tilde{N}$) and the unstable poles of $P$ (which are the unstable zeros of $\tilde{M}$)  relate to the actuator and sensor attacks considered in \cite{dual}.  Multirate sampling can potentially remove unstable zeros of $\tilde{N}_{yd}$ as it was shown in \cite{dual} for unbounded actuator attacks, but it cannot work for total sensor unbounded attacks.

In the following we consider the case of bounded in magnitude (and time) attacks with various levels of stealth.  The question we want to address is how to compute the worst possible bounded attacks and how to defend against such attacks by a suitable controller design.

\section{Computation of Worst Attack}

We consider the problem of computing the worst
case attack in \eqref{eqn:attack} when the attacker has a finite time window $\{0,1, \dots ,t_a\}$ to
attack the system. In addition, we can possibly require the attack to remain stealthy after the attack is over. This allows for repeatedly attacking the system without
triggering monitoring signal alarm.

%
%
%

Specifically, consider the optimization problem in \eqref{eqn:attack}. Assume the LTI closed loop system $\Phi(K)$ is stable and let $t_{\psi d}$ and $t_{z d}$ be design parameters related to the decay rate of the pulse responses of of $\Phi_{z d}$ and $\Phi_{\psi d}$ respectively.  These parameters determine the time windows that the attacker cares for impact and stealthiness respectively. Let $t_d=\max(t_{z d},t_{\psi d}$). Suppose the intruder can only attack the system during a finite interval $\{0,1, \dots ,t_a\}$, with attack magnitude less than or equal to $\alpha$.  Further, assume that the time window of interest in solving \eqref{eqn:attack} is $t_a +t_d$, i.e., the $\ell_\infty $ norms are applied for $\{z(k)\}$ and $\{d(k)\}$ with $k=\{0,\dots,t_a+t_d\}$.

The system of equations governing the output $z$ when subjected to the attack input $d$ for each instance of time are given by

\begin{equation}
\!\begin{aligned}
\colvec[.65]{ z(0)\\
z(1)\\
\vdots\\
z(t_a)\\
z(t_a+1)\\
\vdots\\
z(t_a+t_{zd})\\
}=
\colvec[.65]{	
\Phi_{z d}(0)						&	0														&	0							 									 &	\cdots\\
\Phi_{z d}(1)						&	\Phi_{z d}(0)								&	0						   									&	\cdots\\
\vdots									&	\vdots											&	\vdots						 							&	 \vdots\\
\Phi_{z d}(t_a)					&	\Phi_{z d}(t_a-1)						&		\Phi_{z d}(t_a-2)		 					&	\cdots\\
\Phi_{z d}(t_a+1)				&	\Phi_{z d}(t_a)							&\Phi_{z d}(t_a-1)							  &	\cdots\\
\vdots									&	\vdots											&	\vdots												  &	\vdots\\
\Phi_{z d}(t_a+t_{zd})	&	\Phi_{z d}(t_a+t_{zd}-1) 		&	\Phi_{z d}(t_a+t_{zd}-2) 				&	\cdots\\
}
\colvec[.65]{ d(0)\\
d(1)\\
\vdots\\
d(t_a)\\
0\\
\vdots\\
0\\
}
\end{aligned}
\label{eqn:proof1}
\end{equation}
where $d(k)=0$ for $t> t_a$ since the attack time window belongs to $t \in \{0, \dots ,t_a\}$. The objective is to find the sequence $\{d(k)\}$, $k=\{0,\dots ,t_a\}$ that maximizes $\left\|z\right\|_{\infty}$ such that $\left\|\psi \right\|_{\infty}\leq \theta$. This corresponds to selecting the optimal row in \eqref{eqn:proof1} to be maximized and finding the optimal $d$ that would maximize this row.  In view of the above, the following proposition is obvious.

\begin{proposition}
\label{prop:thfinitelp1}
Problem \eqref{eqn:attack} can be formulated as the following optimization problem for finite attack window $t \in \{0, \dots ,t_a\}$:
\begin{equation}
\!\begin{aligned}
&\max\limits_{d,n\in \{0,1,\dots,t_a + t_{zd}\}} \sum\limits_{k=0}^{n}  \Phi_{zd}(n-k)d(k) \\
&\text{s.t.}~\Big|\sum\limits_{k=0}^{\tau}  \Phi_{\psi d}(\tau -k)d(k)\Big|\leq\theta,~  \tau=0, 1, \dots, t_a + t_{\psi d},\\
&~~~~\left\vert d(k)\right\vert\leq\alpha,~ k=0, 1, \dots, t_a,\\
&~~~~~d(k)=0,~~k=t_a+1,\dots,t_a+t_d.\\
\end{aligned}
\label{eqn:ta1}
\end{equation}
After finding the worst case attack $\hat{d}$, the worst case impact can be obtained by computing $\left\|\Phi_{z d}\hat{d}\right\|_{\infty}$.
\end{proposition}

\begin{remark} 
The objective function looks for the optimal row in the set $\{0, \dots ,t_a+t_{zd}\}$. We can always choose a sufficiently long $t_{zd}$, determined by the decay rate of $\Phi_{zd}$ and the bound $\alpha$ on $d$, to ensure that we capture the worst case $\left\|z\right\|$. 
\end{remark}


\begin{remark} Note that the first set of constraints ensures the monitoring signal $\psi$ is below a threshold level ($\left\|\psi \right\|_{\infty}\leq \theta$) during and after the attack interval.
Since we assume that $\Phi_{\psi d}$ is stable and that $d(k)=0$ for $t > t_a$, if $t_{\psi d}$ is chosen long enough, depending on  the decay rate of  $\Phi_{\psi d}$ and the bound $\alpha$, one can guarantee that $d$ is undetectable for all $t$. Therefore, to guarantee stealthiness for all $t$ it is sufficient to enforce the monitoring constraints up to $t_a+t_{\psi d}$. The last set of constraints ensures the attack is bounded and decays to zero at the end of the attack interval.  
\end{remark}

\begin{remark}
Problem \eqref{eqn:ta1} is LP for a fixed $n$ which can be solved efficiently.  Fixing $n$ transforms the objective function to a linear function under linear (polytopic) constraints.  However, one has to solve (in principle) $t_a +t_{zd}$ LPs.  
  \end{remark}

 We now provide a simple search algorithm to solve Problem \ref{eqn:ta1}:

\begin{algorithm}
\caption{Compute worst attack $\hat{d}$}
\begin{algorithmic}
\State Input $ \Phi_{zd}$, $\Phi_{\psi d}$, $t_a$, $t_{zd}$, $t_{\psi d}$, $t_{d}$, $\theta$ and $\alpha$.
\For {$i=1:t_a+t_{zd}$}
\State Solve
\begin{equation*}
\!\begin{aligned}
&\max\limits_{d_i} \sum\limits_{k=0}^{i}  \Phi_{zd}(i-k)d_i(k) \\
&\text{s.t.}~\Big|\sum\limits_{k=0}^{\tau}  \Phi_{\psi d}(\tau -k)d_i(k)\Big|\leq\theta,~  \tau=0, 1, \dots, t_a + t_{\psi d},\\
&~~~~\left\vert d_i(k)\right\vert\leq\alpha,~ k=0, 1, \dots, t_a,\\
&~~~~~d_i(k)=0,~~k=t_a+1,\dots,t_a+t_d.\\
\end{aligned}
\end{equation*}
\State Compute and store $\left\|\Phi_{z d}d_i\right\|_{\infty}$
\EndFor
\State \textbf{end}
\State Compare $\left\|\Phi_{z d}d_i\right\|_{\infty}$ and determine the worst attack $\hat{d}$.

\end{algorithmic}
\end{algorithm}

In the sequel, we consider certain cases which simplify further the computations.  Specifically, we consider the problem of computing the worst
case attack when the attacker has a finite time window $k=\{0,\dots,t_a\}$ to
attack the system such as in Proposition \ref{prop:thfinitelp1}. However, in this case we assume that the intruder does not mind being detected after the attack is over. The following proposition describes how to construct the optimal $d$.
 
\begin{proposition}
\label{prop:thfinitelp2}
Consider the optimization Problem in \eqref{eqn:attack} with $t_{\psi d}=0$.  Then, its solution can be obtained by solving
\begin{equation}
\!\begin{aligned}
&\max\limits_{d,n\in\{t_a,\dots,t_a+t_{z d}\}} \sum\limits_{k=0}^{n}  \Phi_{zd}(n-k)d(k) \\
&\text{s.t.}~\Big|\sum\limits_{k=0}^{\tau}  \Phi_{\psi d}(\tau -k)d(k)\Big|\leq\theta,~  \tau=0, 1, \dots, t_a ,\\
&~~~~\left\vert d(k)\right\vert\leq\alpha,~ k=0, 1, \dots, t_a,\\
&~~~~~d(k)=0,~~k=t_a+1,\dots,t_a+t_{z d}.\\
\end{aligned}
\label{eqn:ta2}
\end{equation}
\end{proposition}

\begin{proof}
We will prove that the optimal row to be maximized is in the set $\{z(t_a), \dots, z(t_a+t_{zd})\}$.
Let $\hat{d}$ be the worst attack that maximizes $\mu=:\left\|z\right\|_{\infty}$ found by solving for the maximum impact over all the rows of \eqref{eqn:proof1} where the stealthiness constraints are enforced up to $t=t_a$. Assume that $\hat{d}$ was found by maximizing any row before $z(t_a)$ calling it row $i$. Since the stealthiness constraints are imposed only up to $t_a$ and $\Phi(K)$ is LTI, we can delay $\hat{d}$ by $t_a-i$ steps (shift $\hat{d}$ to the right) so that $\left\|z\right\|_{\infty}=\mu$ is achieved by maximizing the row $z(t_a)$ without violating the stealthiness constraints. In addition, we cannot shift the attack beyond $z(t_a)$ since $\hat{d}(k)=0$ for $t > t_a$. As a result, maximizing $\left\|z\right\|_{\infty}$ for $k=\{0,1,\dots,t_a\}$ is equivalent to maximizing $\left\|z\right\|_{\infty}$ for $k=\{t_a\}$. 

\end{proof}

\begin{remark}
\label{rem:rem2}
The optimization problem in \eqref{eqn:ta2} differs from the problem in \eqref{eqn:ta1} in two ways: First, the stealthiness constraints set in \eqref{eqn:ta2} is a subset of the set in \eqref{eqn:ta1}, since in \eqref{eqn:ta2} the objective is to remain stealthy only during the attack interval, where in \eqref{eqn:ta1} the stealthiness condition is enforced at all times. Therefore, the attack designed using Proposition \ref{prop:thfinitelp2} yields worse impact in the $\ell_{\infty}$ sense than the attacks designed using Proposition \ref{prop:thfinitelp1}.  The second difference is in the objective function where in \eqref{eqn:ta1} we have to maximize each row in \eqref{eqn:proof1} to find the worst attack (i.e., $t_a+t_{zd}$ LPs), while in \eqref{eqn:ta2} we only need to maximize the last rows associated with $[z(t_a), \dots, z(t_a+t_{zd})]'$ (i.e., $t_{zd}$ LPs).  An immediate corollary is as follows.
\end{remark}

\begin{corollary}
\label{thfinitelp3}
  Let $t_{zd}=t_{\psi d}=0$, i.e., the attacker cares to inflict maximum damage in the window up to $t_a$ while does not care for stealthiness after $t_a$.  Then, the optimal $d$ is obtained by solving the following LP
\begin{equation}
\!\begin{aligned}
&\max\limits_{d} \sum\limits_{k=0}^{t_a}  \Phi_{zd}(t_a-k)d(k) \\
&\text{s.t.}~\Big|\sum\limits_{k=0}^{\tau}  \Phi_{\psi d}(\tau -k)d(k)\Big|\leq\theta,~  \tau=0, 1, \dots, t_a ,\\
&~~~~\left\vert d(k)\right\vert\leq\alpha,~ k=0, 1, \dots, t_a,\\
\end{aligned}
\label{eqn:ta3}
\end{equation}
\end{corollary}

\begin{remark}
If $\Phi_{\psi d}$ is non-minimum phase and $\alpha$ is not specified, then the optimization problem in Proposition \ref{prop:thfinitelp2} and Corollary \ref{thfinitelp3} will yield unbounded zero dynamics attacks \cite{dual}.

\end{remark}

\section{Controller Design for Resiliency - $K$-$d$ iteration }

In view of the previous discussion, a controller design procedure can be formulated based on LP. In particular, given a desired $\ell_1$ performance level $\gamma$ for attacks $d$,  find $K$ such that $\left\|\Phi_{zd}(K)\right\|_1\le\gamma$, and to ensure that for a given attack level  characterized by $\left\| d\right\|_\infty\le\alpha$, where $\alpha$ is an attack resource parameter, the ``undetected loss'' of the closed loop given by
$$\mu_{\alpha} :=\max_{d}\left\|\Phi_{z d}(K)d\right\|_\infty$$ $$\text{s.t.}~ \left\| \Phi_{\psi d}(K)d\right\|_\infty\leq\theta,~\left\| d\right\|_\infty\le\alpha$$
remains below a desired level $\mu$.  Computing $\mu_{\alpha}$ for a given $K$ corresponds to the problem of computing the worst $d$ of the previous section. A synthesis procedure can be developed by a ``K-d'' type of iteration:
	\begin{itemize}
\item Given $K_i$ with $\left\|\Phi_{zd}(K_i)\right\|=\gamma_i$ find $d_i$ from:
$$\mu_i:=\max_{d}\left\|\Phi_{zd}(K_i)d\right\|_\infty~$$ $$\text{s.t.}~ \left\| \Phi_{\psi d}(K_i)d\right\|_\infty\leq\theta,~\left\| d\right\|_\infty\le\alpha.$$
\item Given $d_i$ find $K_{i+1}$ from:
$$\gamma_{i+1}:=\min_{K}\left\|\Phi_{zd}(K)\right\|_1~\text{s.t.}~ \left\| \Phi_{zd}(K)d_i\right\|_\infty\le\mu_i$$
 \item At each iteration $i$ the problem is a LP with
 $$\gamma_{i}\le \gamma_{i-1}\le \gamma_0,~\mu_i\le\gamma_i\left\|d_i\right\|_\infty,~\left\|d_i\right\|_\infty\le\alpha.$$
 	\end{itemize}
The above formulation guarantees that the upper bound on the attack impact (i.e. $\mu_i$) is non-increasing with each iteration.

\section{Conclusions}
We considered the problem of computing worst case bounded stealthy false data injection attacks for LTI systems. We considered different attack resource constraints and stealthiness intervals. This problem involves a maximization of a convex function subject to convex constraints, and it was shown that it can be cast as a series of LP problems under $\ell_\infty$ framework. We provided a search algorithm to solve the set of LPs. Furthermore, we provided an iterative controller synthesis procedure that alternates between computing worst attacks and designing controllers that enhance performance and minimize the impact of worst attacks.


\begin{thebibliography}{10}
\providecommand{\url}[1]{#1}
\csname url@samestyle\endcsname
\providecommand{\newblock}{\relax}
\providecommand{\bibinfo}[2]{#2}
\providecommand{\BIBentrySTDinterwordspacing}{\spaceskip=0pt\relax}
\providecommand{\BIBentryALTinterwordstretchfactor}{4}
\providecommand{\BIBentryALTinterwordspacing}{\spaceskip=\fontdimen2\font plus
\BIBentryALTinterwordstretchfactor\fontdimen3\font minus
  \fontdimen4\font\relax}
\providecommand{\BIBforeignlanguage}[2]{{%
\expandafter\ifx\csname l@#1\endcsname\relax
\typeout{** WARNING: IEEEtran.bst: No hyphenation pattern has been}%
\typeout{** loaded for the language `#1'. Using the pattern for}%
\typeout{** the default language instead.}%
\else
\language=\csname l@#1\endcsname
\fi
#2}}
\providecommand{\BIBdecl}{\relax}
\BIBdecl
\bibitem{sandberg1}
A.~Teixeira, S.~Amin, H.~Sandberg, K.~H. Johansson, and S.~S. Sastry, ``Cyber
  security analysis of state estimators in electric power systems,'' in
  \emph{49th IEEE Conference on Decision and Control}, December 2010, pp.
  5991--5998.

\bibitem{tabuada}
H.~Fawzi, P.~Tabuada, and S.~Diggavi, ``Secure estimation and control for
  cyber-physical systems under adversarial attacks,'' \emph{IEEE Transactions
  on Automatic Control}, vol.~59, no.~6, pp. 1454--1467, June 2014.

\bibitem{hespanha}
M.~S. Chong, M.~Wakaiki, and J.~P. Hespanha, ``Observability of linear systems
  under adversarial attacks,'' in \emph{2015 American Control Conference
  (ACC)}, July 2015, pp. 2439--2444.

\bibitem{fraz}
S.~Z. Yong, M.~Q. Foo, and E.~Frazzoli, ``Robust and resilient estimation for
  cyber-physical systems under adversarial attacks,'' in \emph{American Control
  Conference}, July 2016, pp. 308--315.

\bibitem{dual}
M.~Naghnaeian, N.~Hirzallah, and P.~G. Voulgaris, ``Dual rate control for
  security in cyber-physical systems,'' in \emph{2015 54th IEEE Conference on
  Decision and Control}, December 2015, pp. 1415--1420.

\bibitem{sinopoli}
Y.~Mo, R.~Chabukswar, and B.~Sinopoli, ``Detecting integrity attacks on scada
  systems,'' \emph{IEEE Transactions on Control Systems Technology}, vol.~22,
  no.~4, pp. 1396--1407, July 2014.

\bibitem{sandberg2}
A.~Teixeira, I.~Shames, H.~Sandberg, and K.~H. Johansson, ``A secure control
  framework for resource-limited adversaries,'' \emph{Automatica}, vol.~51, pp.
  135--148, 2015.

\bibitem{insok}
C.~Kwon, W.~Liu, and I.~Hwang, ``Analysis and design of stealthy cyber attacks
  on unmanned aerial systems,'' \emph{Journal of Aerospace Information
  Systems}, vol.~11, no.~8, pp. 525--539, August 2014.

\bibitem{lqg}
R.~Zhang and P.~Venkitasubramaniam, ``Stealthy control signal attacks in vector
  lqg systems,'' in \emph{2016 American Control Conference}, July 2016, pp.
  1179--1184.

\bibitem{wu}
G.~Wu and J.~Sun, ``Optimal data integrity attack on actuators in
  cyber-physical systems,'' in \emph{2016 American Control Conference (ACC)},
  July 2016, pp. 1160--1164.

\bibitem{antsaklisBook}
P.~Antsaklis and A.~N. Michel, \emph{Linear Systems}.\hskip 1em plus 0.5em
  minus 0.4em\relax McGraw-Hill, 2006.

\bibitem{ZhouDoyleGloverBook}
K.~Zhou, J.~C. Doyle, and K.~Glover, \emph{Robust and Optimal Control}.\hskip
  1em plus 0.5em minus 0.4em\relax Prentice-Hall, 1995.

\bibitem{dahlehb}
M.~A. Dahleh and I.~Diaz-Bobillo, \emph{Control of Uncertain Systems: A Linear
  Programming Approach}.\hskip 1em plus 0.5em minus 0.4em\relax Prentice-Hall,
  1995.

\end{thebibliography}
\end{document}